\documentclass[journal]{IEEEtran}
\usepackage{cite}
\usepackage{tabularx}
\usepackage{amsmath,amssymb,amsfonts}
\usepackage{algorithmic}
\usepackage{graphicx}
\usepackage{textcomp}
\usepackage{xcolor}
\usepackage{url}
\usepackage{bm}
\usepackage{multirow}
\usepackage{placeins} 
\usepackage[ruled,vlined]{algorithm2e}

\allowdisplaybreaks

\begin{document}

\title{Resilience Enhancement of Distribution Grids Through a Three-Stage Framework for Scheduling Mobile Energy Storage Systems}

\author{
Ali Abbasi, Kyri Baker\\
\textit{University of Colorado Boulder}\\
Boulder, CO, USA\\
\{ali.abbasi, kyri\}@colorado.edu
}

\maketitle

\begin{abstract}
As high-impact, low-probability disturbances such as extreme weather events increase in frequency, enhancing power grid resilience has become a critical priority. Mobile Energy Storage Systems (MESSs) have strong potential for helping with this purpose due to their high operational flexibility and fast deployability. They can be rapidly relocated to supply critical loads, support islanded operation, and adapt to changing conditions of the grid, making them a promising addition or alternative to conventional resilience-providing methods. This paper proposes a three-stage framework for scheduling MESS units.
In the first stage (Normal Operation), charging/discharging scheduling of MESSs at buses is optimized for economic arbitrage.
Upon receiving an early warning signal, the second stage (Proactive Positioning) shifts the focus to resilience. Here, the model determines the optimal staging points within the transportation network to minimize the expected time of arrival at critical loads, considering the probability distribution of line outages.
The third stage (Dynamic Relocation) addresses the post-event restoration phase, in which MESSs are relocated in response to grid updates, such as subsequent line failures.
Furthermore, a graph neural network-based solution approach is proposed to further enhance computational efficiency in determining the optimal placing of MESS units during the pre-allocation and relocation stages.
The simulation results demonstrate that the proposed framework reduces expected energy not served and improves the restoration speed.
\end{abstract}

\begin{IEEEkeywords}
Mobile energy storage, grid resilience, graph neural networks, transportation-power nexus, multi-stage optimization
\end{IEEEkeywords}

\section*{Nomenclature}
\label{sec:Nomenclature}

\subsection*{\textit{Sets and Indices}}

\begin{IEEEdescription}[\IEEEsetlabelwidth{$k \in \mathcal{K}$}]
\item[$b \in \mathcal{B}$] Set of buses
\item[$k \in \mathcal{K}$] Set of energy storage units
\item[$l \in \mathcal{L}$] Set of distribution lines
\item[$t \in \mathcal{T}$] Set of time intervals
\item[$s \in \mathcal{S}$] Set of scenarios
\end{IEEEdescription}

\subsection*{\textit{Parameters}}
\begin{IEEEdescription}[\IEEEsetlabelwidth{$C^{\mathrm{VoLL}}_b$}]
\item[$\eta_{\mathrm{ch/dis}}$] Charging/Discharging efficiency
\item[$\pi_s$] Probability of scenario $s$
\item[$r(l)$] Receiving-end bus of line $l$
\item[$s(l)$] Sending-end bus of line $l$
\item[$\overline{E}_k$] Energy rating of energy storage units [MWh]
\item[$\overline{P}_k$] Power rating of energy storage units [MW]
\item[$B_b$] Susceptance of bus $b$ [$\Omega^{-1}$]
\item[$C^{\mathrm{g}}_b$] Incremental cost of distributed generator [\$]
\item[$C^{\mathrm{VoLL}}_b$] Value of lost load [\$/MWh]
\item[$G_b$] Conductance of bus $b$ [$\Omega$]
\item[$K$] Power factor calculation coefficient
\item[$P^{\mathrm{d}}_{bts}$] Real power demand [MW]
\item[$Q^{\mathrm{d}}_{bts}$] Reactive power demand [MVar]
\item[$R_l$] Resistance of distribution line $l$ [$\Omega$]
\item[$S_l$] Apparent flow limit of distribution line $l$ [MVA]
\item[$u_{kb}$] Transit route parameter of ES unit in mobile mode: 1 if ES unit $k$ is located at bus $b$, 0 otherwise
\item[$X_l$] Reactance of distribution line $l$ [$\Omega$]
\end{IEEEdescription}

\vspace{1em}

\subsection*{\textit{Binary Variables}}
\begin{IEEEdescription}[\IEEEsetlabelwidth{$\sigma^{\mathrm{d}}_b$}]
\item[$\sigma^{\mathrm{d}}_b$] Load switch state variable: 1 if load at bus $b$ is connected, 0 otherwise
\item[$\sigma^{\mathrm{l}}_l$] Line switch state variable: 1 if line $l$ is closed, 0 otherwise
\item[$z_{kb}$] Stationary location variable of ES unit: 1 if ES unit $k$ is located at bus $b$, 0 otherwise
\end{IEEEdescription}

\vspace{1em}

\subsection*{\textit{Variables}}
\begin{IEEEdescription}[\IEEEsetlabelwidth{$P^{\mathrm{ch/dis}}_{kbt}$}]
\item[$a_{lt}$] Current flow of distribution line $l$ [p.u.]
\item[$e_{kt}$] Energy state-of-charge [MWh]
\item[$f^{\mathrm{p/q}}_l$] Real/reactive power flow of line $l$ [MW/MVar]
\item[$P^{\mathrm{ch/dis}}_{kbt}$] Charging/discharging real power decision [MW]
\item[$P^{\mathrm{g}}_{bt}$] Real power output of distributed generator [MW]
\item[$q^{\mathrm{ch/dis}}_{kbt}$] Charging/discharging reactive power decision [MVar]
\item[$q^{\mathrm{g}}_{bts}$] Reactive power output of distributed generator [MVar]
\item[$v_{bt}$] Nodal voltage magnitude of bus $b$ [p.u.]
\end{IEEEdescription}

\section{Introduction}
\label{sec:introduction}

\IEEEPARstart{D}{istribution} networks (DNs) are particularly vulnerable to high-impact, low-probability disturbances such as hurricanes, floods, wildfires, and other extreme events, which can lead to widespread customer outages \cite{shi2022enhancing}. As these threats have increased in both intensity and frequency \cite{paul2024resilience}, resilience has become a central objective in the planning and operation of DNs, which unlike reliability, measures the ability to continue to serve critical loads and minimize unserved energy in low-probability extreme events \cite{paul2024resilience}. Enhancing the resilience of DNs, therefore, becomes essential for reducing the societal and economic impacts of prolonged outages.

Mobile energy storage systems (MESSs), typically implemented as battery storage units mounted on trucks or trailers, have emerged as a promising resilience solution for modern DNs due to their ability to combine the fast response of battery energy storage with physical mobility across the network \cite{chuangpishit2023mobile,dugan2021application}. Unlike conventional stationary storage, which is tied to its installation site, MESSs can be transported to vulnerable areas before an extreme event, be dispatched to outage locations after damage occurs, and be repositioned during the restoration process as system conditions evolve, which makes them especially attractive for resilience enhancement. In addition, MESSs can provide value during normal operating conditions, such as performing spatio-temporal energy arbitrage, which helps justify their investment beyond rare emergency events \cite{kim2018enhancing}.

\subsection{Literature Review}
\hspace*{0.25in}Previous works have explored a range of modeling, scheduling, and application aspects of MESSs in resilience-enhancement \cite{dugan2021application,lu2024mobile,aslam2025application}.
Specific scenarios are considered, like earthquake \cite{rajabzadeh2022improving}, ice disaster \cite{guo2023mobile}, typhoon \cite{zhou2024bi}, rainfalls \cite{hua2023robust}, or the scheduling is solved for the general case. However, much of the prior work remains limited to the post-disaster restoration-only scheduling, where the objective is usually to minimize load shedding. In these models, MESS routing and charging/discharging are co-optimized with other emergency actions such as network reconfiguration, microgrid dispatch, repair scheduling, or critical-load restoration. For example, in \cite{yao2019rolling}, Yao \textit{et al.} modeled MESS operation using a stochastic multi-layer time-space network and solved a rolling two-stage stochastic MILP for coordinated restoration. Nazemi \textit{et al.} developed a stochastic optimization model with joint probabilistic constraints in \cite{9372331}, later reformulated as a tractable MILP, for routing and scheduling MESS units jointly with dynamic reconfiguration and stochastic renewable energy sources. In \cite{shen2023mobile}, Shen \textit{et al.} proposed a bilevel post-disaster recovery model in which the upper level minimizes total load curtailment and the lower level minimizes voltage offset during recovery.
While exploring post-disaster recovery provides valuable insights into this stage, focusing solely on post-disaster recovery ignores the interconnected stages of MESS scheduling. As successful restoration depends on earlier actions like pre-positioning, neglecting the full timeline leads to inefficient MESS utilization, slower responses, and reduced overall resilience.

Multi-stage MESS scheduling has also been widely explored in the literature, such as \cite{10311549,9634033,app142210367,10049743}. In these works, the problem is typically formulated as a multi-stage optimization task, where the locations, movements, and charging/discharging decisions of MESSs are coordinated across several stages—such as preventive positioning, emergency operation, and restoration—to minimize interruption following outages. 
In \cite{10311549} for example, Chen \textit{et al.} proposed a two-stage stochastic mixed-integer programming (SMIP) model, with the first stage formulating the planning decisions for MESSs, and the second stage evaluating the operating costs of the DN under normal, severe, and extreme scenarios. 
A multi-stage optimization method for MESS scheduling in a coupled transportation and distribution network has been proposed in \cite{9634033}. In this work, authors use real-time traffic information and update routing and charging/discharging decisions of MESSs under uncertainties such as photovoltaic (PV) penetration; however, the main focus is on PV accommodation and economic dispatch rather than resilience-driven post-disaster restoration.
Lei \textit{et al.} optimized capacity sizing of MESSs in \cite{app142210367} as an often neglected aspect of scheduling of these units. This feature is optimized along with pre-positioning of MESSs before a natural disaster in the first stage of the problem formulation, while in the post-disaster second stage, re-allocation of MESSs and their power outputs are considered to be adjusted based on the damage scenarios.
A customer-centric three-stage MESS planning framework was presented \cite{10049743}, that decomposes the problem into pre-outage battery delivery, in-outage battery relocation, and post-outage battery recycling, focusing on direct residential backup service with multiple carriers.
These formulations represent a clear improvement over single-stage models as they capture part of the temporal coupling. However, they still often remain focused on the disaster-response horizon only, 
While it is more economically efficient to employ MESSs during normal operation as well to, avoid leaving them idle until an emergency occurs and they are dispatched to assigned buses.

\subsection{Contributions}
Previous works have considered incorporating MESSs under normal conditions \cite{kim2018enhancing}; however, we can further introduce an additional level of optimality to the overall scheduling. An intermediate stage between pre-outage and post-outage MESS scheduling is often overlooked: the urgency of delivering supply to critical loads immediately following an outage. The key challenge is that addressing this urgency requires allocating MESSs on the transportation network (TN) level, while to be able to contribute during normal operation, they have to be located at DN buses. In other words, MESSs are first located at buses during normal operation, and are later dispatched to buses for post-outage support; however, there is a phase in between, during which, they need to be staged on the transportation roads, so that their expected time of arrival (ETA) at critical buses would be minimal. 
Additionally, the dynamic nature of post-outage grid conditions is often overlooked. If MESS schedules are not continuously adapted to real-time changes—such as secondary damages, they can only provide sub-optimal power, resulting in an inefficient use of these resources.
Considering these additional stages, the problem can now be formulated as a three-stage optimization framework. Figure~\ref{fig:flowchart} shows a flow chart of transitions between these stages.
Furthermore, to enable sufficiently rapid responses to system updates during the relocation stage, and the evaluation of multiple scenarios in the pre-allocation stage and identifying the optimal strategy, we propose a deep learning (DL)-based approach utilizing on graph neural networks (GNNs) to reduce computational burden and improve responsiveness.

The main contributions of this paper are summarized as follows:
\begin{itemize}
    \item We propose a strategy that bridges the gap between daily normal operation and emergency preparedness. The proposed framework utilizes MESSs for price arbitrage during normal grid conditions and shifts to an optimized pre-allocation strategy on the TN upon receiving warning signals.

    \item To address the computational intractability of the formulated problem for the pre-allocation stage, we propose an algorithm to decouple the original problem into three computationally solvable sub-problems.
    
    \item We develop a dynamic response mechanism that continuously monitors the post-outage environment. This strategy accounts for the evolving nature of the system during outages, including updates on repair progress of damaged lines, the emergence of secondary contingencies, and variations in state-of-charge (SoC) of MESSs.

    \item To address the computational burden of traditional optimization methods, we propose a Deep Learning-based solution. By shifting the heavy computation of the conventional solution to an offline training phase, we leverage the fast inference capabilities of neural networks to map grid states to optimal dispatch and control decisions. This provides near-instantaneous responses to rapid network updates that might not be computationally possible using standard solvers.
\end{itemize}

The remainder of this paper is organized as follows. Section II details the multi-stage mathematical problem formulation, including the normal operation stage, pre-allocation on the transportation network, and active relocation. Section III presents the DL-based solution architecture. The numerical simulation results are discussed in Section IV. Finally, Section V concludes the paper and provides potential directions for future research.

\begin{figure}[t]
    \centering
    \includegraphics[width=0.95\linewidth]{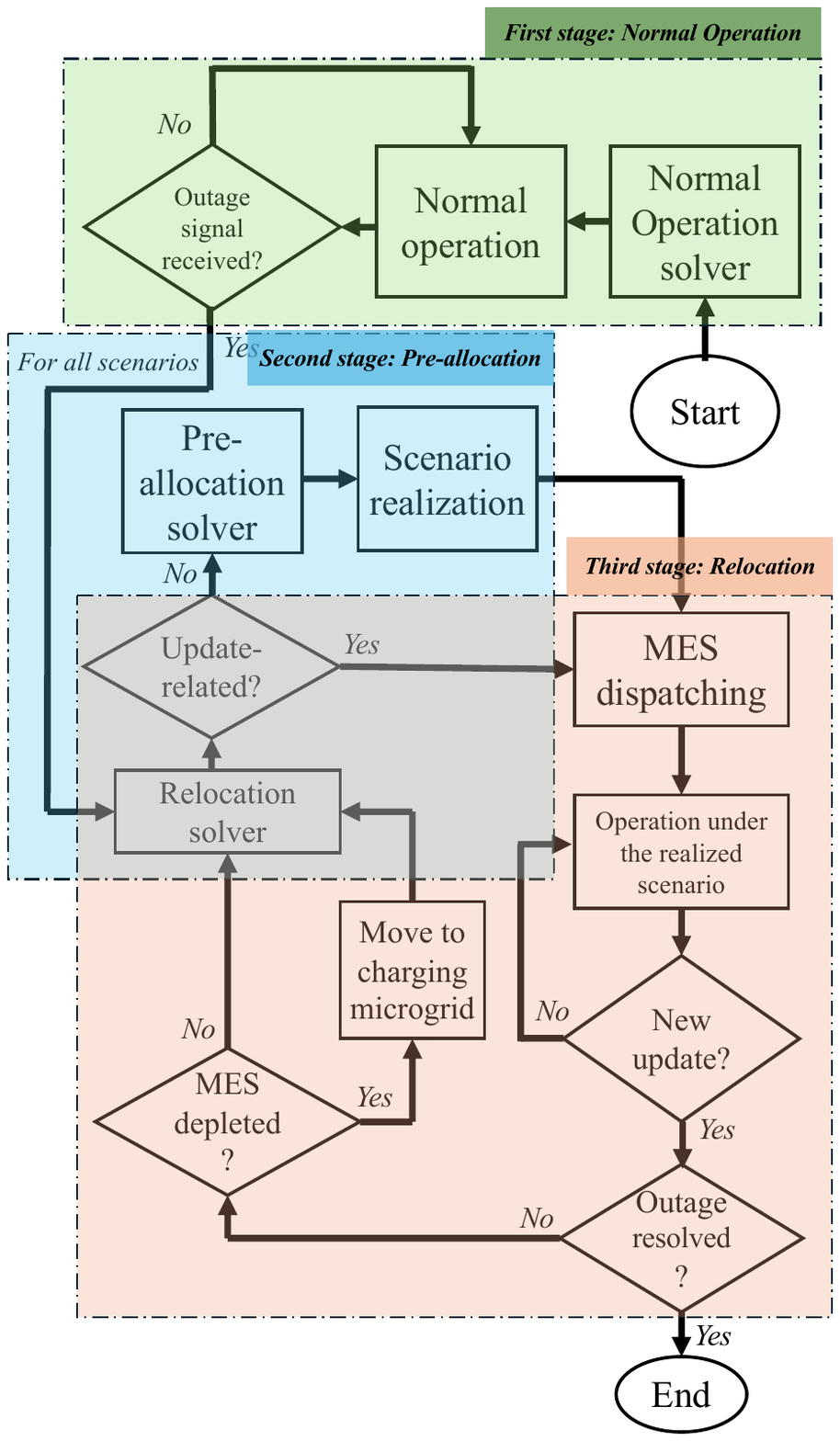}
    \caption{Flow chart of transitions between different stages.}
    \label{fig:flowchart}
\end{figure}

\section{Problem Formulation and the Solution}
\label{section: Problem Formulation and Solution}

In this section, we present the problem formulation for the proposed three-stage MESS scheduling framework and its solution. We assume a radial grid topology, as it is the most common architecture in distribution systems \cite{ahmadi2015mathematical}, coupled with a transportation network. Our problem formulation includes three stages, providing a comprehensive model of the system's operation. During the first stage, i.e., Normal Operation, MESS units are fixed at bus locations to perform energy arbitrage. Once a warning signal for an outage within a grid zone is received, the system transitions to the Pre-allocation stage. Here, MESS units are moved to optimal positions within the transportation network to minimize expected travel times. Finally, during the Dynamic Relocation stage, the DSO continuously monitors the system and relocates MESS units in response to real-time grid updates, such as new fault occurrences or repairs. Fig. 1 shows a flow chart of transitions between the three stages.

\subsection{Normal Operation}
During normal operation, the DSO utilizes MESS units to perform energy arbitrage. These units are assumed to be at fixed locations at buses at this stage, as the transportation costs and operational downtime associated with relocation would outweigh the marginal economic gains from relocating MESS units during this stage. Consequently, the goal of this stage is to determine the optimal placement of MESS units and their optimal charging and discharging schedules at fixed locations.

The objective function for this stage is as follows:
\begin{equation}
\label{eq:normal_operation_objective}
\min \sum_{t \in T, b \in B} C_{bt}^{g} p_{bt}^{g}
\end{equation}
The objective function for this stage only aims to minimize the cost of the generated power by generators. The decision variables are the placements of MESS units, which are restricted to buses, and their charging/discharging rates at different times.

The distribution system constraints at this stage are as follows:
\begin{subequations}
\begin{flalign}
& (f^p_{lt})^2 + (f^q_{lt})^2 \leq S_{l}^{2}, \quad \forall l \in \mathcal{L} & \label{eq:normal_operation_constraints_a} \\
& (f^p_{lt} - a_{lt}^2 \cdot R_l)^2 + (f^q_{lt} - a_{lt}^2 \cdot X_l)^2 \leq S_l^2, \quad \forall l \in \mathcal{L} & \label{eq:normal_operation_constraints_b} \\
& v_{s(l),t}^2 - 2(R_{l} \cdot f_{lt}^{p} + X_{l} \cdot f_{lt}^{q}) + a_{lt}^2(R_{l}^{2} + X_{l}^{2}) \nonumber & \\
& = v_{r(l),t}^2, \quad \forall l \in \mathcal{L} & \label{eq:normal_operation_constraints_c} \\
& \frac{(f_{lt}^{p})^{2} + (f_{lt}^{q})^{2}}{a_{lt}^2} \leq v_{s(l),t}^2, \quad \forall l \in \mathcal{L} & \label{eq:normal_operation_constraints_d} \\
& - \sum_{l|s(l)=0} (f^p_{lt} - a_{lt}^2 \cdot R_l) - p^g_{0,t} + G_0 \cdot v_{0,t}^2 = 0 & \label{eq:normal_operation_constraints_e} \\
& - \sum_{l|s(l)=0} (f^q_{lt} - a_{lt}^2 \cdot X_l) - q^g_{0,t} - B_0 \cdot v_{0,t}^2 = 0 & \label{eq:normal_operation_constraints_f} \\
& f^p_{bt} - \sum_{l|r(l)=b} (f^p_{lt} - a_{lt}^2 \cdot R_l) - p^g_{bt} + P^d_{bt} + G_b \cdot v_{bt}^2 \nonumber & \\
& \quad \quad - \sum_{k \in \mathcal{K}} p^{\text{dis}}_{kbt} + \sum_{k \in \mathcal{K}} p^{\text{ch}}_{kbt} = 0, \quad \forall b \in \mathcal{B} & \label{eq:normal_operation_constraints_g} \\
& f^q_{bt} - \sum_{l|r(l)=b} (f^q_{lt} - a_{lt}^2 \cdot X_l) - q^g_{bt} + Q^d_b - B_b \cdot v_{bt}^2 \nonumber & \\
& \quad \quad - \sum_{k \in \mathcal{K}} q^{\text{dis}}_{kbt} + \sum_{k \in \mathcal{K}} q^{\text{ch}}_{kbt} = 0, \quad \forall b \in \mathcal{B} & \label{eq:normal_operation_constraints_h} \\
& \underline{P}^g_b \leq p^g_{bt} \leq \overline{P}^g_b, \quad \forall b \in \mathcal{B}^G & \label{eq:normal_operation_constraints_i} \\
& \underline{Q}^g_b \leq q^g_{bt} \leq \overline{Q}^g_b, \quad \forall b \in \mathcal{B}^G & \label{eq:normal_operation_constraints_j} \\
& \underline{V}_{b} \leq v_{bt} \leq \overline{V}_{b}, \quad \forall b \in \mathcal{B} & \label{eq:normal_operation_constraints_k}
\end{flalign}
\end{subequations}

This set of constraints formulates the physical operations of the radial distribution grid using the Second-Order Cone Programming (SOCP) relaxation of the branch flow (DistFlow) equations \cite{6507355}. This model enforces physical constraints while remaining computationally tractable; and is known to be exact for radial topologies under a set of mild, unrestrictive assumptions \cite{6815671}. In case of meshed networks, other AC power flow models such as LinDistFlow \cite{25627} could be used instead. Specifically, the formulation enforces nodal power balance for both active and reactive power (Eqs. \eqref{eq:normal_operation_constraints_e}–\eqref{eq:normal_operation_constraints_h}), making the power flowing into a bus match the power consumed by loads, injected by generators, or exchanged (charged/discharged) by MESS units, at the root bus (Eqs. (\eqref{eq:normal_operation_constraints_e}-\eqref{eq:normal_operation_constraints_f}), or other buses (Eqs. \eqref{eq:normal_operation_constraints_g}-\eqref{eq:normal_operation_constraints_h}). Critical safety limits are also enforced in \eqref{eq:normal_operation_constraints_a} and \eqref{eq:normal_operation_constraints_b} for thermal capacity limits of lines, \eqref{eq:normal_operation_constraints_k} for voltage magnitude limits of buses, and \eqref{eq:normal_operation_constraints_i} and \eqref{eq:normal_operation_constraints_j} for generation capacity limits, while \eqref{eq:normal_operation_constraints_c} and \eqref{eq:normal_operation_constraints_d} link voltage drops and current flows via convex relaxation.

The following set of constraints model MESS units.
\begin{subequations}
\begin{flalign}
& e_{kt} = e_{k,t-1,} + \sum_{b \in \mathcal{B}} (p_{kbt}^{\text{ch}} \cdot \eta_{\text{ch}} - p_{kbt}^{\text{dis}}/\eta_{\text{dis}}) & \label{eq:mes_norm_constraints_a} \\
& 0 \leq e_{kt} \leq \overline{E}_k & \label{eq:mes_norm_constraints_b} \\
& e_{k,t_0} = e_{k,t_{24}} = 0.5 \overline{E}_k  & \label{eq:mes_norm_constraints_c} \\
& 0 \leq p_{kbt}^{\text{ch}} \leq \overline{P}_k \cdot z_{kb} \cdot (1 - \mathbb{I}(p_{kbt}^{\text{dis}} > 0)), \quad \forall b \in \mathcal{B} & \label{eq:mes_norm_constraints_d} \\
& 0 \leq p_{kbt}^{\text{dis}} \leq \overline{P}_k \cdot z_{kb} \cdot (1 - \mathbb{I}(p_{kbt}^{\text{ch}} > 0)), \quad \forall b \in \mathcal{B} & \label{eq:mes_norm_constraints_e} \\
& \sum_{b \in \mathcal{B}} z_{kb} = 1, \quad \forall k \in \mathcal{K} & \label{eq:mes_norm_constraints_f} \\
& -K \cdot p_{kbt}^{\text{ch}} \leq q_{kbt}^{\text{ch}} \leq K \cdot p_{kbt}^{\text{ch}}, \quad \forall b \in \mathcal{B} & \label{eq:mes_norm_constraints_g} \\
& -K \cdot p_{kbt}^{\text{dis}} \leq q_{kbt}^{\text{dis}} \leq K \cdot p_{kbt}^{\text{dis}}, \quad \forall b \in \mathcal{B} & \label{eq:mes_norm_constraints_h}
\end{flalign}
\end{subequations}

Battery state of charge dynamics are defined in \eqref{eq:mes_norm_constraints_a}, while capacity limits are enforced in \eqref{eq:mes_norm_constraints_b}. Equation \eqref{eq:mes_norm_constraints_c} ensures that the initial and final State of Charge (SoC) for the units are constrained to be equal, which serves as a boundary condition for the scheduling horizon. We enforce in \eqref{eq:mes_norm_constraints_d} and \eqref{eq:mes_norm_constraints_e} power rating limits for charging and discharging; and also that an MESS unit cannot inject or absorb power from a bus unless it is physically located there by multiplying the binary location variable $z_{kb}$ by $\overline{P}_k$. By further multiplying that by $(1 - \mathbb{I}(p_{kbt}^{\text{ch/dis}} > 0))$, we prevent charging and discharging happening at the same time, where $\mathbb{I}(\cdot)$ stands for the indicator function. Eq. \eqref{eq:mes_norm_constraints_f} ensures that each MESS unit should be located at exactly one bus. Constraints \eqref{eq:mes_norm_constraints_g} and \eqref{eq:mes_norm_constraints_h} couple the reactive power limits of the MESS units to their active power dispatch with the power factor parameter $K$.

By using the SOCP relaxation of the power flow equations, the overall problem ends up in the form of a Mixed-Integer Second-Order Cone Program (MISOCP). The problem could then be implemented and solved using solvers such as Gurobi, which efficiently handles the discrete integer decision varibales alongside the continuous network constraints.

\subsection{Pre-allocation}
When a warning signal for a potential disturbance—such as a severe weather event is received, the system transitions from normal operation to the pre-allocation stage. In this phase, MESS units move from their arbitrage-based fixed locations at buses to pre-allocated points on the TN. Unlike the deterministic nature of normal operation, this stage is affected by the uncertainty of which scenario is about to happen. The warning signal is modeled to contain a set of potential damage scenarios, $\mathcal{S}$, where each scenario $s \in \mathcal{S}$ is associated with a probability of occurrence, $\pi_s$. In practice, these parameters can be obtained by passing weather forecasts through component fragility curves \cite{7801854}, \cite{9372331}, or by evaluating cyber-intrusion alerts using Bayesian attack networks \cite{ibne2020modeling}. Each scenario is characterized by a subset of grid lines and transportation roads that may be affected by the disturbance.

The optimization problem in this stage is twofold. First, for every possible damage scenario $s$, we find the optimal restoration locations for the MESS units within the grid. This is obtained by solving a load-shedding minimization problem (detailed in Section II.C) to identify the specific bus $b_{s,i}$ where MESS unit $i$ must connect to to maximize the saved load if scenario $s$ occurred. Once the target buses for all scenarios are identified, the objective is to find a single holding location $(x_i, y_i)$ for each MESS unit on the transportation network. These locations are chosen such that they minimize the ETA at the designated target buses, weighted by the probability of each scenario.

\begin{figure}[!t]
    \centering
    \includegraphics[width=\columnwidth]{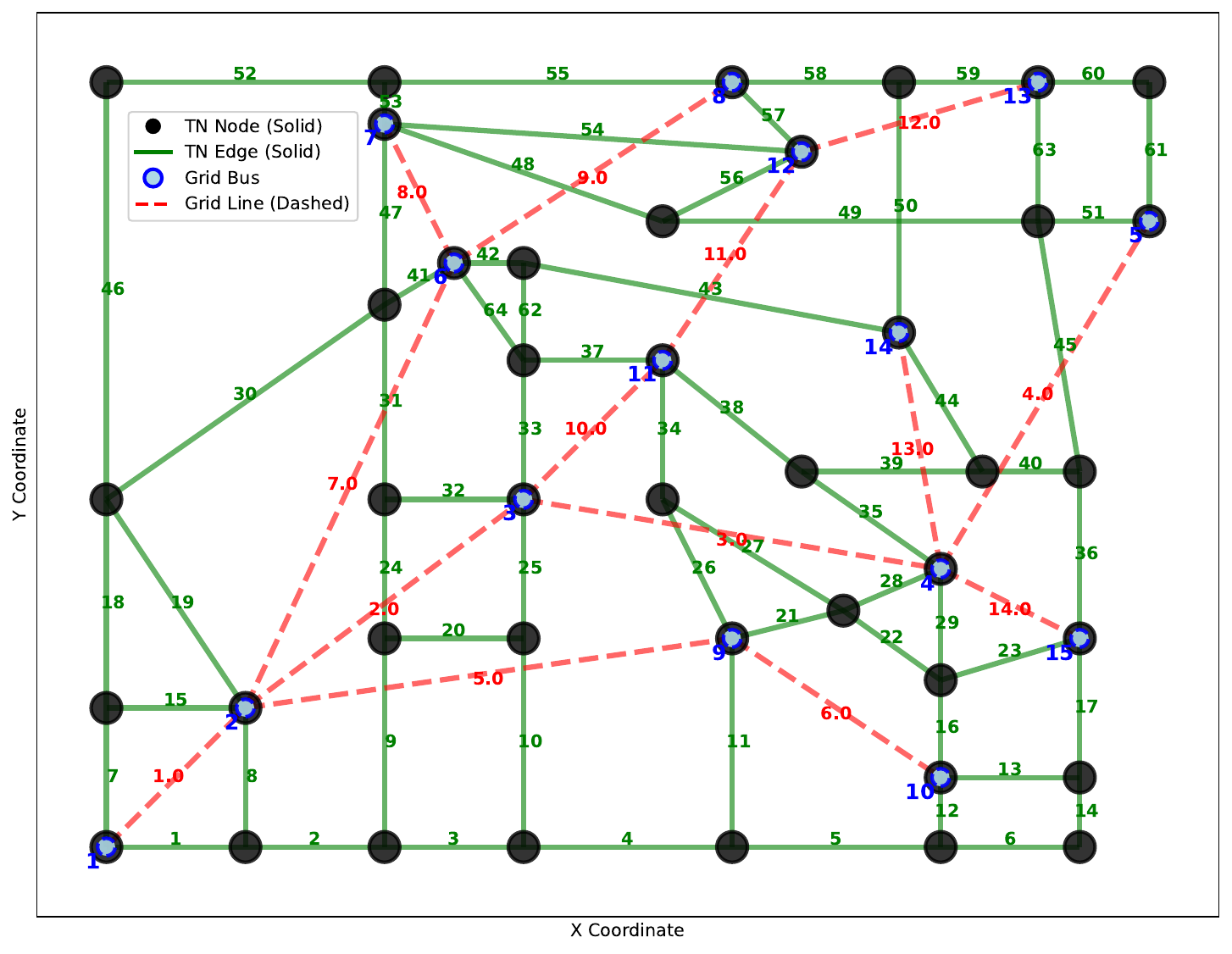}
    \caption{The overlay of the transportation network and the grid. To ensure road accessibility, each bus is assumed to be located at a transportation node.}
    \label{fig:TN_grid_overlay}
\end{figure}

The system restoration plan cannot begin until all of the required units have arrived at the designated buses. Therefore, for a given scenario $s$, the effective response time is governed by the last arriving unit. The objective function for this stage minimizes the expected value of this arrival time:
\begin{equation}
\label{eq:minmax_objective}
\min \sum_{s \in \mathcal{S}} \pi_s \cdot \left( \max_{i \in \mathcal{M}} (t_{i,s}) \right) \quad
\end{equation}
where $t_{i,s}$ denotes the travel time for MESS unit $i$ to reach its designated bus $b_{s,i}$ under scenario $s$.

The transportation network model is simplified as a graph $\mathcal{G} = (\mathcal{V}, \mathcal{E})$, where $\mathcal{V}$ denotes the set of vertices (nodes) and $\mathcal{E}$ represents the set of edges. A depiction of the modeled transportation network and the overlaid grid is illustrated in Fig. \ref{fig:TN_grid_overlay}.

To address the computational intractability of the joint stochastic min-max problem presented in \eqref{eq:minmax_objective}, we propose a sequential heuristic solution. Algorithm 1 decouples the original problem into three computationally solvable sub-problems: scenario evaluation, spatial clustering, and discrete stochastic node search.

\begin{algorithm}[!t]
\DontPrintSemicolon
\SetKwInOut{Input}{Input}
\SetKwInOut{Output}{Output}
\SetKw{KwTo}{to}

\caption{Pre-allocation of MES Units via Clique Clustering and Stochastic Node Search}
\Input{Scenarios $\mathcal{S}$, Probabilities $\{\pi_s\}$, \\
       Transportation Graph $\mathcal{G} = (\mathcal{V},\mathcal{E})$ with \\ shortest path distances $d(u,v)$, \\
       Grid Topology, $K$ MES Units}
\Output{Optimal pre-allocating locations $\mathcal{L}^* = \{v_1^*, \dots, v_K^*\}$}

\textbf{Step 1. Identify Optimal Placements per Scenario} \\
\For{$s \leftarrow 1$ \KwTo $|\mathcal{S}|$}{
    Solve grid optimization (minimize lost load cost): \\
    $\mathcal{B}_s \leftarrow \arg\min_b \sum_{b \in \mathcal{B}} C_b^{\text{VoLL}}(1 - \sigma_b^d)P_b^d$ \\
    \textit{Result:} $\mathcal{B}_s = \{b_{s,1}, \dots, b_{s,K}\}$ is the set of optimal assigned buses for scenario $s$.
}

\textbf{Step 2. Form Clusters via Clique Cost} \\
Partition the scenario sets $\mathcal{B}_1, \dots, \mathcal{B}_{|\mathcal{S}|}$ into $K$ disjoint clusters $C_1, \dots, C_K$. \\
\textbf{Subject to:} Each cluster $C_k$ must contain exactly one node from each scenario set $\mathcal{B}_s$. \\
\textbf{Objective:} Minimize the sum of all pairwise shortest path distances within the clusters: \\
\hspace{2em} $\min \sum_{k=1}^{K} \sum_{\substack{u,v \in C_k \\ u \neq v}} d(u, v)$ \\
\textit{Result:} $K$ optimized clusters, with $b_{s,k}^*$ denoting the node assigned to cluster $C_k$ from scenario $s$. \\

\textbf{Step 3. Stochastic Node Search} \\
\For{$k \leftarrow 1$ \KwTo $K$}{
    Initialize best cost $J^* \leftarrow \infty$ and best location $v_k^* \leftarrow \text{null}$. \\
    \For{$v \in \mathcal{V}$}{
        $J(v) = \sum_{s \in \Omega} \pi_s \cdot d(v, b_{s,k}^*)$ \\
        \If{$J(v) < J^*$}{
            $J^* \leftarrow J(v)$ \\
            $v_k^* \leftarrow v$
        }
    }
}
\textbf{return} $\mathcal{L}^* = \{v_1^*, \dots, v_K^*\}$
\end{algorithm}

The first stage of the decomposition (Step 1) isolates the distribution grid constraints from the transportation model. By solving the load shedding minimization problem for each disturbance scenario $s \in \Omega$ independently, a set of optimal ``target buses'' $\mathcal{B}_s$ for each scenario is generated. With these optimal bus sets, we now aim to cluster buses that are closer to each other, so we can then `solve the problem for each cluster separately. We can, therefore, formulate this problem as a Multidimensional Assignment Problem with decomposable pairwise costs, commonly referred to as the Clique Partitioning problem \cite{grotschel1989cutting}. We can express this in the following form:

\begin{equation}
    \label{eq:clique_objective}
    \min_{x \in \{0, 1\}} \sum_{k=1}^{K} \sum_{s_1<s_2}^{|\mathcal{S}|} \sum_{i=1}^{K} \sum_{j=1}^{K} d(b_{s_1,i}, b_{s_2,j}) \cdot x_{s_1, i, k} \cdot x_{s_2, j, k}
\end{equation}
\begin{equation}
\begin{aligned}
\label{eq:clique_constraint}
\text{s.t.} \quad
\sum_{k=1}^{K} x_{s,i,k} = 1,
\quad \forall s \in \mathcal{S},
\; \forall i \in \{1,\ldots,K\}.
\end{aligned}
\end{equation}

The objective function \eqref{eq:clique_objective} minimizes the total global clique cost, which leads to clustering nodes that share close proximity on the transportation graph. This proximity is measured by real paths on the transportation network, accounting for road closures, rather than Euclidean distances. Here, $d(b_{s_1,i}, b_{s_2,j})$ denotes the shortest path distance between nodes $b_{s_1,i}$ and $b_{s_2,j}$ through the transportation network, which can be calculated using Djkstra's algorithm. The binary variable $x_{s, i, k}$ equals 1 if the $i$-th candidate bus from scenario $s$ is assigned to cluster $k$, and 0 otherwise. Constraint \eqref{eq:clique_constraint} enforces that each candidate node from each scenario should be assigned to exactly one cluster. To formulate this problem as a Mixed-Integer Linear Programming (MILP) problem, we can introduce the additional binary auxiliary variable $y_{s_1, i, s_2, j, k} = x_{s_1,i,k} \cdot x_{s_2,j,k}$ to linearize the objective function \eqref{eq:clique_objective}. For all state pairs where $s_1 < s_2$, and indices $i, j, k \in \{1, \dots, K\}$, the problem becomes: 

\begin{equation}
\label{eq:clique_objective_linear}
\min_{x,y} \sum_{k=1}^K \sum_{s_1 < s_2}^{|\mathcal{S}|} \sum_{i=1}^K \sum_{j=1}^K d(b_{s_1,i}, b_{s_2,j}) y_{s_1,i,s_2,j,k}
\end{equation}
\begin{subequations}
\begin{align}
\label{eq:clique_constraint_linear}
\text{s.t.} \quad & \sum_{k=1}^K x_{s,i,k} = 1, \\
& y_{s_1,i,s_2,j,k} \le x_{s_1,i,k}, \\
& y_{s_1,i,s_2,j,k} \le x_{s_2,j,k}, \\
& y_{s_1,i,s_2,j,k} \ge x_{s_1,i,k} + x_{s_2,j,k} - 1, \\
& x_{s,i,k} \in \{0, 1\}, \\
& y_{s_1,i,s_2,j,k} \in \{0, 1\}.
\end{align}
\end{subequations}

Once the clusters are formed, we need to find the optimal placement of the MESS units on the transportation network for each cluster. The following lemma proves that this problem can be reduced to searching nodes only.

\newtheorem{lemma}{Lemma}

\begin{lemma}
\label{lem:hakimi}
Let $\mathcal{G} = (\mathcal{V}, \mathcal{E})$ be a continuous transportation graph. For any cluster $C_k$ containing target nodes $\{b_{1,k}^*, b_{2,k}^*, \dots, b_{|\mathcal{S}|,k}^*\}$ generated under scenario $s \in \mathcal{S}$ with the corresponding probability $\pi_s \geq 0$, the optimal point $p^* \in \mathcal{G}$ that minimizes the expected shortest distance to all nodes in $C_k$ is a vertex of the graph, $p^* \in \mathcal{V}$.
\end{lemma}

\textit{Proof.} Let $p(\lambda)$ be an arbitrary point located on the edge $(u, v) \in \mathcal{E}$, with $\lambda$ being the distance from $p(\lambda)$ to node $u$, $0 \leq \lambda \leq w_{uv}$, and $w_{uv}$ being the length of this edge. The distance from $p(\lambda)$ to node $v$ is then $w_{uv} - \lambda$.

For any target node $b_{s,k}^* \in C_k$, the shortest path from $p(\lambda)$ to $b_{s,k}^*$ must pass through either $u$ or $v$:
\begin{equation}
    d(p(\lambda), b_{s,k}) = \min \{ \lambda + d(u, b_{s,k}^*), (w_{uv} - \lambda) + d(v, b_{s,k}^*) \}
    \label{eq:edge_distance}
\end{equation}

Equation \eqref{eq:edge_distance} is the minimum of two linear functions with respect to $\lambda$. The point-wise minimum of affine functions is a concave function. Therefore, $d(p(\lambda), b_{s,k}^*)$ is concave on the closed interval $[0, w_{uv}]$.

The objective function is the expected distance from $p(\lambda)$ to all targets in cluster $C_k$:
\begin{equation}
    J(p(\lambda)) = \sum_{s \in \Omega} \pi_s \cdot d(p(\lambda), b_{s,k}^*)
    \label{eq:expected_distance}
\end{equation}

Because $\pi_s \geq 0$, $J(p(\lambda))$ is a non-negative linear combination of concave functions. Therefore, $J(p(\lambda))$ is itself a concave function on the interval $[0, w_{uv}]$.
A fundamental property of a concave function defined on a closed interval is that its global minimum must occur at one of the extreme points of the interval. For the interval $[0, w_{uv}]$, these extreme points are $\lambda = 0$ (associated with node $u$) and $\lambda = w_{uv}$ (associated with node $v$). Generalizing this to all other edges, we conclude that the optimal solution $p^*$ must be in the set of vertices $\mathcal{V}$. \hfill $\blacksquare$

Once $v_i^*$ has been found for each cluster, the set $\mathcal{L}^* = \{v_1^*, \dots, v_K^*\}$ is fully determined.

\subsection{Relocation}
The third stage of the proposed framework addresses the stochastic evolution of the distribution network during the active outage period. Unlike many other works that assume static allocations, this stage operates in a dynamic environment where the DSO must continuously adapt to changes in grid topology and resource availability. The solutions from this stage are also essential for the optimal dispatching in the pre-allocation stage. We categorize the updates that trigger a relocation event into three types: 1) an MESS unit’s State of Charge (SoC) drops below a  threshold or a new MESS is fully charged  and ready to be deployed, 2) a previously damaged distribution line is repaired and returned to service, 3) a new distribution line is disconnected due to expanding fault conditions. The system response is determined by the nature of the update. In the case of charge depletion, the depleted MESS unit is routed to a microgrid containing an active generator, and remains there until fully charged. Updates related to topological changes trigger a re-optimization of the network. Whenever a line status changes (repair or failure), the restoration problem is re-optimized to determine the updated optimal bus locations for all available MESS units. 

The objective function in this stage is as follows.
\begin{equation}
\label{eq:reloc_objective}
\min \sum_{b \in \mathcal{B}} C_b^{\text{VoLL}} (1 - \sigma_b^d) P_b^d
\end{equation}

The focus here is on minimizing load-shedding; therefore, the objective function in \eqref{eq:reloc_objective} is designed to minimize the cost of the lost load over the entire buses. For each bus, the priority of its load is modeled in the coefficient $C_b^{VoLL}$. The term $(1 - \sigma_b^d)$ determines if the load at bus $b$, $P_b^d$, is served or shed using the binary decision variable, $\sigma_b^d$. Other decision variables for this stage are the open/closed status of distribution lines, $\sigma_l^l$, to form islanded microgrids, and the placement of MESS units, $u_{kb}$.

The distribution system constraints for this stage are:
\begin{subequations}
\begin{flalign}
& \left| v_{s(l)}^2 - v_{r(l)}^2 - 2(R_l f_{l}^p + X_l f_{l}^q) + a_{l}^2(R_l^2 + X_l^2) \right| \leq \nonumber & \\
& M \cdot (1 - \sigma_{l}^l), \quad \forall l \in \mathcal{L} & \label{eq:reloc_constraints_a} \\
& (f_{l}^{p})^{2} + (f_{l}^{q})^{2} \leq v_{s(l)}^2{a_{l}^2}, \quad \forall l \in \mathcal{L} & \label{eq:reloc_constraints_b} \\
& (f^{\mathrm{p}}_{l})^2 + (f^{\mathrm{q}}_{l})^2 \leq (\sigma^{\mathrm{l}}_{l}S_l)^2, \quad \forall l \in \mathcal{L} & \label{eq:reloc_constraints_c} \\
& (f^{\mathrm{p}}_{l} - a_{l}^2 \cdot R_l)^2 + (f^{\mathrm{q}}_{l} - a_{l}^2 \cdot X_l)^2 \leq (\sigma^{\mathrm{l}}_{l} \cdot S_l)^2, \quad \forall l \in \mathcal{L} & \label{eq:reloc_constraints_d} \\
& - \sum_{l|s(l)=0} (f_{l}^{p} - a_{l}^2 \cdot R_{l}) - p_{0}^{g} + G_{0} \cdot v_{0}^2 = 0 & \label{eq:reloc_constraints_e} \\
& - \sum_{l|s(l)=0} (f_{l}^{q} - a_{l}^2 \cdot X_{l}) - q_{0}^{g} - B_{0} \cdot v_{0}^2 = 0 & \label{eq:reloc_constraints_f} \\
& f^{\mathrm{p}}_{b} - \sum_{l|r(l)=b} (f^{\mathrm{p}}_{l} - a_{l}^2 \cdot R_l) - p^{\mathrm{g}}_{b} + \sigma^{\mathrm{d}}_{b} \cdot P^{\mathrm{d}}_{b} + G_b \cdot v_{b}^2 \nonumber & \\
& \quad - \sum_{k \in \mathcal{K}} p^{\mathrm{dis}}_{kb}  = 0, \quad \forall b \in \mathcal{B} & \label{eq:reloc_constraints_g} \\  
& f^{\mathrm{q}}_{b} - \sum_{l|r(l)=b} (f^{\mathrm{q}}_{l} - a_{l}^2 \cdot X_l) - q^{\mathrm{g}}_{b} + \sigma^{\mathrm{d}}_{b} \cdot Q^{\mathrm{d}}_{b} - B_b \cdot v_{b}^2 \nonumber & \\
& \quad - \sum_{k \in \mathcal{K}} q^{\mathrm{dis}}_{kb} = 0, \quad \forall b \in \mathcal{B} & \label{eq:reloc_constraints_h} \\
& \underline{P}_{b}^{g} \leq p_{b}^{g} \leq \overline{P}_{b}^{g}, \quad \forall b \in \mathcal{B}^{G} & \label{eq:reloc_constraints_i} \\
& \underline{Q}_{b}^{g} \leq q_{b}^{g} \leq \overline{Q}_{b}^{g}, \quad \forall b \in \mathcal{B}^{G} & \label{eq:reloc_constraints_j} \\
& \underline{V}_{b} \leq v_{b} \leq \overline{V}_{b}, \quad \forall b \in \mathcal{B} & \label{eq:reloc_constraints_k}
\end{flalign}
\end{subequations}

Similar to the previous stage, these constraints model the power flow via SOCP relaxation. However, binary decision variables are introduced here to control topology changes and load management. Additionally, the formulation in this stage relies on a time-independent, single-period formulation that calculates the optimal response for the current network conditions once, and the variables are recalculated only upon receiving an update. The Big-M formulation in \eqref{eq:reloc_constraints_a} enforces the physical voltage drop along a line when it is closed ($\sigma_l^l = 1$), and decouples the sending and receiving node voltages when the ine is open ($\sigma_l^l = 0$). Using the load shedding variable, $\sigma_b^d$, nodal balance is enforced in \eqref{eq:reloc_constraints_g} and \eqref{eq:reloc_constraints_h} for when the load is connected or shed.

The MESS constraints for this stage are also modeled as follows:
\begin{subequations}
\begin{flalign}
& 0 \leq p^{\mathrm{dis}}_{kb} \leq \bar{P}_k \cdot u_{kb}, \quad \forall k \in \mathcal{K}, b \in \mathcal{B} & \label{eq:mes_reloc_constraints_a} \\
& \sum_{b \in \mathcal{B}} u_{kb} = 1, \quad \forall k \in \mathcal{K} & \label{eq:mes_reloc_constraints_b} \\
& -K \cdot p_{kb}^{\text{dis}} \leq q_{kb}^{\text{dis}} \leq K \cdot p_{kb}^{\text{dis}}, \quad \forall b \in \mathcal{B} & \label{eq:mes_reloc_constraints_c}
\end{flalign}
\end{subequations}
Unlike the normal operation stage, the MESS units can only operate in the discharging mode in this stage. Additionally, the time-independent nature of this formulation allows for the removal of all temporal coupling constraints.

The problem formulation for this stage remains as an MISOCP, and can be solved via conventional solvers. However, the extensive solution times of traditional solvers make them hard to employ for such use cases, as real-time updates need to be responded to timely, and for obtaining the optimal solutions for the pre-allocation stage, we need to process many scenarios in a reasonably quick time. For this purpose, deep learning-based solutions show great potential, which will be discussed in the next section.

\section{DL-based Solution}
\label{section: Deep Learning-based Solution}

Conventional branch-and-cut solvers require excessive computation time to solve the load-shedding minimization problem. This difficulty becomes even more severe in the pre-allocation stage, where numerous problems of this kind must be solved in a short time period to find the optimal placements. As a result, adopting more computationally efficient approaches becomes necessary. 

\subsection{Graph Neural Networks}
\label{subsec:gnn}

A power network can be represented as a graph
\(\mathcal{G}=(\mathcal{V},\mathcal{E})\), where
\(\mathcal{V}=\{1,\dots,N\}\) denotes the set of buses and
\(\mathcal{E}\subseteq \mathcal{V}\times\mathcal{V}\) denotes the set of transmission lines.
Unlike Euclidean data, graph-structured data are defined on an irregular domain whose dependencies are governed by the network topology. Graph neural
networks (GNNs) are designed to learn on such irregular domains by iteratively propagating information along the graph \cite{battaglia2018relational}. Since the present problem involves edge-related disturbances and requires both node-level and
edge-level decisions, an edge-aware message-passing neural
network (MPNN) is adopted \cite{pmlr-v70-gilmer17a}.

For each node \(i\in\mathcal{V}\), let
\(\mathbf{x}_i \in \mathbb{R}^{d_v}\) denote its input feature vector, and for each edge
\((i,j)\in\mathcal{E}\), let
\(\mathbf{e}_{ij}\in\mathbb{R}^{d_e}\) denote its edge feature vector.
In an MPNN, each node \(i\) is associated with a hidden state
\(\mathbf{h}_i^{(\ell)} \in \mathbb{R}^{d_h}\), and each edge \((i,j)\) is associated with an edge state
\(\mathbf{z}_{ij}^{(\ell)} \in \mathbb{R}^{d_z}\) at layer \(\ell\).
The initial embeddings are obtained as:
\begin{equation}
\mathbf{h}_i^{(0)} = \phi_v^{\mathrm{in}}(\mathbf{x}_i), \qquad
\mathbf{z}_{ij}^{(0)} = \phi_e^{\mathrm{in}}(\mathbf{e}_{ij}),
\label{eq:init_embed}
\end{equation}
where \(\phi_v^{\mathrm{in}}(\cdot)\) and \(\phi_e^{\mathrm{in}}(\cdot)\) are learnable input encoders.

Then, for \(\ell=0,1,\dots,L-1\), information is propagated according to four operations: message construction, neighborhood aggregation, node update, and edge update.

First, for each ordered pair \((j,i)\) such that \(j\in\mathcal{N}(i)\), where \(\mathcal{N}(i)\) is the neighbor set of node \(i\), a message is computed as:
\begin{equation}
\mathbf{m}_{j\rightarrow i}^{(\ell)}
=
\phi_m^{(\ell)}
\!\left(
\mathbf{h}_i^{(\ell)},
\mathbf{h}_j^{(\ell)},
\mathbf{z}_{ij}^{(\ell)}
\right),
\label{eq:message_construction}
\end{equation}
where
\(\phi_m^{(\ell)}(\cdot)\) is the message function.
The incoming messages are then aggregated through a permutation-invariant operator:
\begin{equation}
\bar{\mathbf{m}}_i^{(\ell)}
=
\bigoplus_{j\in\mathcal{N}(i)}
\mathbf{m}_{j\rightarrow i}^{(\ell)},
\label{eq:message_aggregation}
\end{equation}
where \(\bigoplus\) is typically chosen as summation.
This step combines the information received from all neighboring nodes. The node embedding is then updated according to:
\begin{equation}
\mathbf{h}_i^{(\ell+1)}
=
\phi_h^{(\ell)}
\!\left(
\mathbf{h}_i^{(\ell)},
\bar{\mathbf{m}}_i^{(\ell)}
\right),
\qquad i\in\mathcal{V},
\label{eq:node_latent_update}
\end{equation}
where \(\phi_h^{(\ell)}(\cdot)\) is the node update function.
Hence, the new representation of node \(i\) depends on both its previous state and the aggregated information received from its neighbors.

Finally, the state of each edge is updated as:
\begin{equation}
\mathbf{z}_{ij}^{(\ell+1)}
=
\phi_z^{(\ell)}
\!\left(
\mathbf{h}_i^{(\ell)},
\mathbf{h}_j^{(\ell)},
\mathbf{z}_{ij}^{(\ell)}
\right),
\qquad (i,j)\in\mathcal{E},
\label{eq:edge_latent_update}
\end{equation}
where \(\phi_z^{(\ell)}(\cdot)\) is the edge update function. The new edge representation \(\mathbf{z}_{ij}^{(L)}\) captures both local line attributes and the states of the surrounding buses.

The final states are then passed to task-specific output heads.
For a node-level task \(r\), the prediction at bus \(i\) is written as:
\begin{equation}
\hat{y}_i^{(r)}
=
\rho_v^{(r)}\!\left(\mathbf{h}_i^{(L)}\right),
\qquad i\in\mathcal{V},
\label{eq:node_decoder}
\end{equation}
where \(\rho_v^{(r)}(\cdot)\) is a learnable node decoder.
Similarly, for an edge-level task \(q\), the prediction on line \((i,j)\) is:
\begin{equation}
\hat{y}_{ij}^{(q)}
=
\rho_e^{(q)}
\!\left(
\left[
\mathbf{h}_i^{(L)}
\;\|\;
\mathbf{h}_j^{(L)}
\;\|\;
\mathbf{z}_{ij}^{(L)}
\right]
\right),
\qquad (i,j)\in\mathcal{E},
\label{eq:edge_decoder}
\end{equation}
where \(\rho_e^{(q)}(\cdot)\) is a learnable edge decoder and \(\|\) denotes concatenation.
For binary decision variables, the decoder output is typically followed by a sigmoid activation to produce a Bernoulli probability.

The above MPNN formulation is particularly suitable for the problem presented in this work.
Unlike node-only architectures, it includes line states in the information flow through the network and produces embeddings for both bus-level and line-level decisions.

\subsection{GNN-Based Solution to the Relocation Problem}
\label{subsec:gnn_relocation}

In the load-shedding minimization problem, the objective is to learn a mapping from the operating conditions of the grid to the optimal MESS placements and other control decisions.
Let
\(\Xi\) denote the collection of all static physical parameters of the grid, such as line impedance, line capacity, and bus admittance quantities.
Since both the topology and the physical parameters are fixed, the proposed network is trained for a specific grid and learns a grid-dependent decision rule.
However, the operating conditions may vary from sample to sample through the input features.

For a given input sample, the features consist of:
(i) the set of damaged lines,
(ii) the values of lost load at buses,
(iii) the real power demands, and
(iv) the reactive power demands at buses.
Accordingly, the node feature vector of bus \(b\in\mathcal{V}\) is defined as:
\begin{equation}
\mathbf{x}_b
=
\left[
C_b^{\mathrm{VoLL}},
P_b^{d},
Q_b^{d}
\right]^\top
\in\mathbb{R}^{3},
\label{eq:reloc_node_feature}
\end{equation}
while the edge feature of line \(l\in\mathcal{E}\) is defined by the binary damage indicator:
\begin{equation}
\mathbf{e}_l = [\delta_l] \in \mathbb{R}, \qquad
\delta_l =
\begin{cases}
1, & \text{if line } l \text{ is damaged},\\
0, & \text{otherwise}.
\end{cases}
\label{eq:reloc_edge_feature}
\end{equation}
Therefore, the learned model can be written as:
\begin{equation}
f_{\Theta}:
\left(
\mathcal{G},
\Xi,
\{\mathbf{x}_b\}_{b\in\mathcal{V}},
\{\mathbf{e}_l\}_{l\in\mathcal{E}}
\right)
\mapsto
\left(
\{u_{kb}\},
\{\sigma_l^{l}\},
\{\sigma_b^{d}\}
\right),
\label{eq:reloc_mapping}
\end{equation}
where \(\Theta\) denotes the trainable parameters of the network.

Using the MPNN model introduced in the previous subsection, after \(L\) message-passing layers, we obtain
\begin{equation}
\left\{
\mathbf{h}_b^{(L)}, \mathbf{z}_l^{(L)}
\right\}
=
\Phi_{\Theta_{\mathrm{enc}}}
\!\left(
\mathcal{G},
\Xi,
\{\mathbf{x}_b\}_{b\in\mathcal{V}},
\{\mathbf{e}_l\}_{l\in\mathcal{E}}
\right),
\label{eq:reloc_encoder}
\end{equation}
where
\(\mathbf{h}_b^{(L)} \in \mathbb{R}^{d_h}\) is the final state of bus \(b\),
\(\mathbf{z}_l^{(L)} \in \mathbb{R}^{d_z}\) is the final state of line \(l\), and
\(\Theta_{\mathrm{enc}}\) is the parameter set of the graph encoder.
These embeddings collect both the local attributes of each bus/line and the nonlocal effect of the damaged-line configuration propagated through the network topology.

The final layer of the model consists of three task-specific output heads corresponding to MESS locations, line switching, and load shedding.
For the location decision, a score is computed for each MESS--bus pair \((k,b)\).
Let \(\mathbf{g}_k \in \mathbb{R}^{d_u}\) denote a learnable embedding associated with MESS unit \(k\).
The relocation score is then defined as:
\begin{equation}
s_{kb}^{u}
=
\rho_u
\!\left(
\left[
\mathbf{g}_k \,\|\, \mathbf{h}_b^{(L)}
\right]
\right),
\label{eq:reloc_score}
\end{equation}
where \(\rho_u(\cdot)\) is a learnable decoder.
A bus-assignment probability for MESS unit \(k\) is obtained through a softmax operation over all buses:
\begin{equation}
p_{kb}^{u}
=
\frac{\exp(s_{kb}^{u})}
{\sum_{b' \in \mathcal{V}} \exp(s_{kb'}^{u})}.
\label{eq:reloc_softmax}
\end{equation}
The binary relocation decision is then obtained as:
\begin{equation}
\hat{u}_{kb}
=
\mathbb{I}
\!\left(
b = \arg\max_{b' \in \mathcal{V}} p_{kb'}^{u}
\right),
\label{eq:reloc_output}
\end{equation}
where \(\mathbb{I}(\cdot)\) is the indicator function.

For the line switching decision, an edge-level decoder is applied to the final latent representation of each line together with the latent states of its incident buses.
Using \(s(l)\) and \(r(l)\) to denote the sending-end and receiving-end buses of line \(l\), respectively, define:
\begin{equation}
s_{l}^{\sigma}
=
\rho_{\sigma}
\!\left(
\left[
\mathbf{h}_{s(l)}^{(L)}
\,\|\,
\mathbf{h}_{r(l)}^{(L)}
\,\|\,
\mathbf{z}_{l}^{(L)}
\right]
\right).
\label{eq:switch_score}
\end{equation}

The predicted switching status of line \(l\) is then:
\begin{equation}
\hat{\sigma}_{l}^{l}
=
\mathbb{I}(\mathrm{sigmoid}(s_{l}^{\sigma})>\tau_{\sigma}),
\label{eq:switch_output}
\end{equation}
where \(\tau_{\sigma}\in(0,1)\) is a decision threshold.

Similarly, for the load shedding decision, a node-level decoder is used:
\begin{equation}
s_{b}^{d}
=
\rho_{d}
\!\left(
\mathbf{h}_{b}^{(L)}
\right),
\label{eq:shed_score}
\end{equation}

\begin{equation}
\hat{\sigma}_{b}^{d}
=
\mathbb{I}(\mathrm{sigmoid}(s_{b}^{d})>\tau_{d}),
\label{eq:shed_output}
\end{equation}
where \(\tau_d\in(0,1)\) is the threshold associated with the load-shedding classifier.

The network is trained in a supervised manner using optimal solutions generated offline from the conventional optimization solver.
Let
\(\mathcal{D}=\{(\mathcal{X}^{(n)},\mathcal{Y}^{(n)})\}_{n=1}^{N_{\mathrm{tr}}}\)
denote the training set, where \(\mathcal{X}^{(n)}\) contains the input features of sample \(n\), and
\begin{equation}
\mathcal{Y}^{(n)}
=
\left(
\{u_{kb}^{\star,(n)}\},
\{\sigma_{l}^{l,\star,(n)}\},
\{\sigma_{b}^{d,\star,(n)}\}
\right)
\label{eq:training_label}
\end{equation}
contains the corresponding optimal placements, switching, and shedding decisions.
The trainable parameters are obtained by minimizing a multi-task loss:
\begin{equation}
\mathcal{L}(\Theta)
=
\lambda_u \mathcal{L}_u
+
\lambda_{\sigma} \mathcal{L}_{\sigma}
+
\lambda_d \mathcal{L}_d,
\label{eq:reloc_total_loss}
\end{equation}
where \(\lambda_u\), \(\lambda_{\sigma}\), and \(\lambda_d\) are nonnegative weighting coefficients.
A natural choice is to define:
\begin{equation}
\mathcal{L}_u
=
-\sum_{n=1}^{N_{\mathrm{tr}}}
\sum_{k}
\sum_{b}
u_{kb}^{\star,(n)}
\log p_{kb}^{u,(n)},
\label{eq:loss_u}
\end{equation}

and

\begin{equation}
\begin{aligned}
\mathcal{L}_{\sigma}
= - \sum_{n=1}^{N_{\mathrm{tr}}} \sum_{l}
\Big[
& \sigma_{l}^{\,l,\star,(n)} \log p_{l}^{\sigma,(n)} \\
&+ \big(1-\sigma_{l}^{\,l,\star,(n)}\big)
\log\!\big(1-p_{l}^{\sigma,(n)}\big)
\Big].
\end{aligned}
\label{eq:loss_sigma}
\end{equation}

$\mathcal{L}_{d}$ can be defined similarly to $\mathcal{L}_{\sigma}$.
It is worth noting that the proposed model does not aim to learn a universal policy across arbitrary network topologies.
Rather, because \(\mathcal{G}\) and \(\Xi\) are fixed, the trained MPNN acts as a high-speed decision model for a given distribution system.
Within that fixed system, however, it can accommodate varying conditions through changes in the damaged-line set and bus-level operating features.

Once the network has been trained for the load-shedding minimization problem, it can be employed in the pre-allocation stage. We can replace the optimization problem solved in Step 2 of Algorithm 1 with a forward pass of the trained GNN. This substitution is particularly valuable in that stage, where a large number of problems corresponding to each scenario should be solved in a short time to determine the optimal pre-allocation of MESS units.

\section{Simulation Results}
\label{section: Simulation Results}

To validate the performance of the proposed three-stage framework, comprehensive case studies are conducted. The proposed methodology is implemented on an IEEE 15-bus radial distribution test system. The base voltage and apparent power of the system are set at $12.66$ kV and $1$ MVA, respectively, establishing a base impedance of $Z_{base} = 160.28 \ \Omega$. Voltage limits are constrained between $\underline{V} = 0.95$ p.u. and $\overline{V} = 1.05$ p.u. during normal operation, while the lower bound is relaxed to $0.85$ p.u. during emergency islanding configurations. The MESS units integrated into the coupled electrical and transportation network feature a maximum energy capacity of $E^{max} = \text{1.0}$ MWh, an active power rating of $P^{max} = \text{0.5}$ MW, and operate with a charging/discharging efficiency of $\eta = \text{0.95}$. For routing over the TN, an average travel speed of $v^{tr} = \text{40}$ km/h is assumed. Furthermore, to accurately capture the selective load shedding behavior in the restoration stage, the VoLL is heterogeneously assigned across the network, prioritizing critical loads over the others.

\subsection{Normal Operation Simulations}

The simulations for this stage focus on demonstrating the effectiveness of MESS utilization in energy arbitrage. Table \ref{tab:mes_cost} shows the total generation cost for the three selected bus systems under four cases of no MESS, one MESS, two MESSs, and three MESSs. As expected, increasing the number of MESS units reduces the total generation cost, with the most significant reduction observed between the cases of no MESSs and a single MESS unit. As the number of utilized MESSs increases, this gap reduces. The same results hold for different bus systems with almost proportional gaps.

The mechanism underlying the cost reduction is illustrated in detail in Fig. \ref{fig:arbitrage24}. This figure illustrates the temporal variation of generated electricity alongside the corresponding electricity price for the IEEE 15-bus system. As observed, during periods of higher electricity prices, generators produce less electricity due to increased reliance on stored energy in MESSs. Conversely, during periods of lower generation costs, electricity production increases, partly to supply energy for storage in MESSs, which is later utilized when prices are higher.

\begin{table}[!t]
\caption{Total Generation Cost under Various Bus System Cases and Various Number of MES units (\$)}
\label{tab:mes_cost}
\centering
\renewcommand{\arraystretch}{1.1}
\begin{tabular}{lcccc}
\hline
\textbf{System} & \textbf{No MES} & \textbf{1 MES} & \textbf{2 MES} & \textbf{3 MES} \\
\hline
IEEE 15-bus  & 1,211 & 1,057 & 968 & 948 \\
IEEE 33-bus  & 3,730 & 3,412 & 3,255 & 3,167 \\
IEEE 69-bus  & 6,101 & 5,743 & 5,414 & 5,236 \\
\hline
\end{tabular}
\end{table}

\begin{figure}[!t]
    \centering
    \includegraphics[width=\columnwidth]{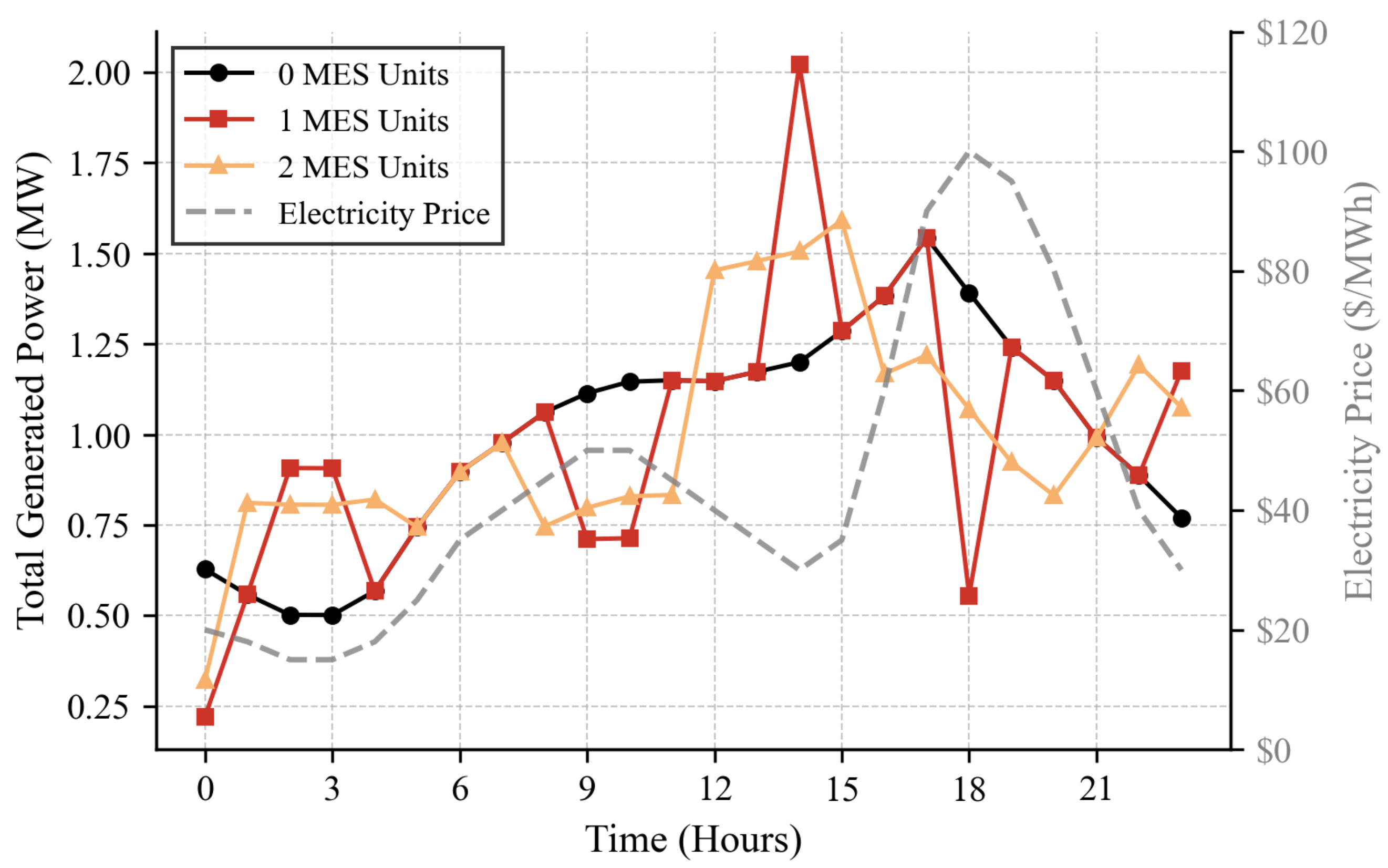}
    \caption{Detailed arbitrage mechanism in a 24-hour Interval.}
    \label{fig:arbitrage24}
\end{figure}

\subsection{Pre-allocation Simulations}

To demonstrate the steps of the proposed algorithm, several scenarios are considered in Table \ref{tab:scenarios}. Each scenario includes the set of failed lines, the assigned MESS locations at buses for that scenario, and the corresponding probability of occurrence. The set of assigned locations are obtained by solving the load-shedding minimization problem formulated in stage III. Figure 4.a shows these assigned bus locations for each scenario. These locations are clustered according to the step II of algorithm I, and the resulting clusters are shown in figure 4.b. It is worth noting that the number of clusters ends up to be equal to the number of MESS units. That is because the number of assigned bus locations for each scenario is equal to the number of MESS units since that is the outcome of solving the load-shedding problem for that scenario. Due to the fact that each cluster contains exactly one bus location per scenario, and that clusters are disjoint, the number of clusters will be equal to the number of MESS units. Fig 4.c shows the nodes assigned on the TN for each cluster according to step 3 of algorithm 1. These selected nodes are those that have the shortest expected path to the assigned bus locations of their clusters.

Table~\ref{tab:prealoc_comparison} compares the ETA and Cost of Lost Load (CoLL) for the cases without pre-allocation and with pre-allocation during the lost period for the IEEE 15-bus and IEEE 33-bus systems. In terms of ETA, pre-allocation reduces this metric by roughly 33\% compared to the case without pre-allocation. As for CoLL, the results indicate that pre-allocation can reduce this by approximately the same ratio as for ETA. Similar trends can be observed for the IEEE 33-bus system.

\begin{table}[!t]
\caption{Failed lines in different scenarios and their corresponding assigned buses}
\label{tab:scenarios}
\centering
\renewcommand{\arraystretch}{1.1}
\begin{tabular}{lcccc}
\hline
\textbf{Scenario} & \textbf{Set of failed lines} & \textbf{set of assigned buses} & \textbf{Probability} \\
\hline
Scenario 1  & \{4,14\} & \{5,15\} & 0.25 \\
Scenario 2  & \{11,3\} & \{4,13\} & 0.25 \\
Scenario 3  & \{11,6\} & \{10,36\} & 0.25 \\
Scenario 4  & \{9,5\} & \{14,39\} & 0.25 \\
\hline
\end{tabular}
\end{table}

\begin{figure}[!t]
    \centering
    \includegraphics[
        width=\columnwidth,
        height=0.78\textheight,
        keepaspectratio
    ]{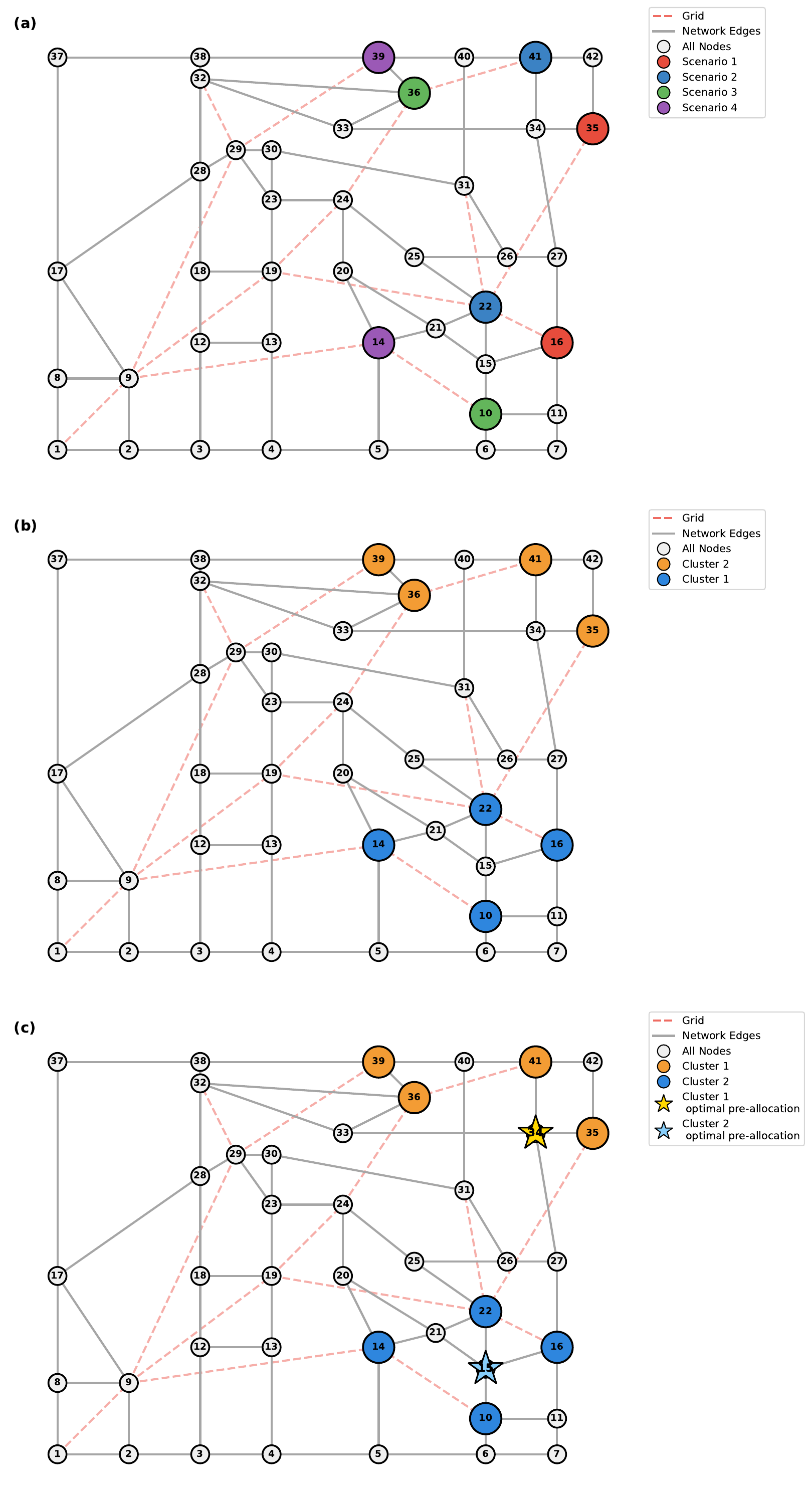}
    \caption{Three steps of Algorithm~1: optimal allocation for each scenario, cluster formation based on clique cost, and stochastic node search.}
    \label{fig:three_plots}
\end{figure}

\begin{table*}[!t]
\caption{Comparison of restoration time and CoLL with and without pre-allocation for different numbers of MES units}
\label{tab:prealoc_comparison}
\centering
\renewcommand{\arraystretch}{1.15}
\setlength{\tabcolsep}{6pt}
\begin{tabular}{llccccc}
\hline
\textbf{System} & \textbf{Metric} & \textbf{0 MES} & \textbf{1 MES} & \textbf{2 MES} & \textbf{3 MES} & \textbf{No Pre-Allocation} \\
\hline
\multirow{2}{*}{IEEE 15} 
& ETA (s) & 2.1 & 1.5 & 1.3 & 1.2 & 2.3 \\
& CoLL (\$)       & 340 & 300 & 280 & 270 & 350 \\
\hline
\multirow{2}{*}{IEEE 33} 
& ETA (s) & 1.9 & 1.6 & 1.55 & 1.5 & 2.2 \\
& CoLL (\$)       & 410 & 370 & 352 & 349 & 430 \\
\hline
\end{tabular}
\end{table*}

\subsection{Relocation Simulations}

The simulations in this section are designed to examine the evolution of system conditions under different update types and the associated control responses to them, including load shedding and MESS relocation decisions. Table~\ref{tab:reloc_results} presents a comprehensive sequence of possible updates in a 12-hour outage interval considering three available MESS units. Fig.~\ref{fig:15bus} illustrates the bus and line indexing used in Table~\ref{tab:reloc_results}. Node 1 is assumed to be connected to the main grid. The system starts with an initial set of damaged lines of \{6,14,3,10\}, and a set of optimal assigned bus locations for MESSs. As the outage progresses, line 4 is added to the set of damaged lines, and MESS locations are updated according to that. Once an MESS unit is depleted, indicated by a $0\%$ energy level at the end of the interval in Table~\ref{tab:reloc_results}, it is moved to the nearest node in the microgrid that is connected to the main grid, while the remaining MESS units are relocated to compensate for its absence. As the outage further evolves, recharged MESS units return to service and some damaged lines are restored until the system reaches the “Outage Resolved” state. Table~\ref{tab:CoLL} illustrates the aggregated VoLL rates and the total CoLL for each of these intervals. The results indicate that the cost is reduced when additional MESSs are available for system support or when the number of damaged lines is lower.

Figure~\ref{fig:voll_evolution} compares the aggregated VoLL rates with and without MESS relocation. Additionally, the impact of number of available MESS units is examined. Figure~\ref{fig:voll_evolution} shows that, as additional lines become damaged, the MESS locations remain fixed in the absence of relocation, resulting in higher aggregated VoLL rates. In contrast, proper relocation allows MESSs to adapt to changes and reduces this rate to nearly zero.

Table~\ref{tab:gnn_accuracy} summarizes the prediction performance of the proposed GNN model for MESS placement, load shedding, and line-status decisions for the IEEE 15, 33, and 69-bus systems. The results show high accuracy for all three categories, which demonstrates its capability to predict the optimal solution.

\begin{table*}[!t]
\caption{Prediction Accuracy of the GNN Model for MESS Placement, Load Shedding, and Line Status on the Training and Validation Sets}
\label{tab:gnn_accuracy}
\centering
\footnotesize
\renewcommand{\arraystretch}{1.15}
\setlength{\tabcolsep}{10pt}

\begin{tabular}{lcccccc}
\hline
\multirow{2}{*}{\textbf{System}}
& \multicolumn{2}{c}{\textbf{MESS Location Accuracy}}
& \multicolumn{2}{c}{\textbf{Load Shedding Accuracy}}
& \multicolumn{2}{c}{\textbf{Line Status Accuracy}} \\
\cline{2-3}
\cline{4-5}
\cline{6-7}
& \textbf{Training}
& \textbf{Validation}
& \textbf{Training}
& \textbf{Validation}
& \textbf{Training}
& \textbf{Validation} \\
\hline
IEEE 15-bus
& 86.2\%
& 82.2\%
& 99.1\%
& 97.2\%
& 99.4\%
& 98.6\% \\

IEEE 33-bus
& 91.1\%
& 85.6\%
& 99.5\%
& 98.1\%
& 99.6\%
& 98.1\% \\

IEEE 69-bus
& 91.7\%
& 86.0\%
& 99.6\%
& 98.3\%
& 99.6\%
& 98.7\% \\
\hline
\end{tabular}

\end{table*}

\begin{figure}[htbp]
    \centering
    \includegraphics[width=\columnwidth]{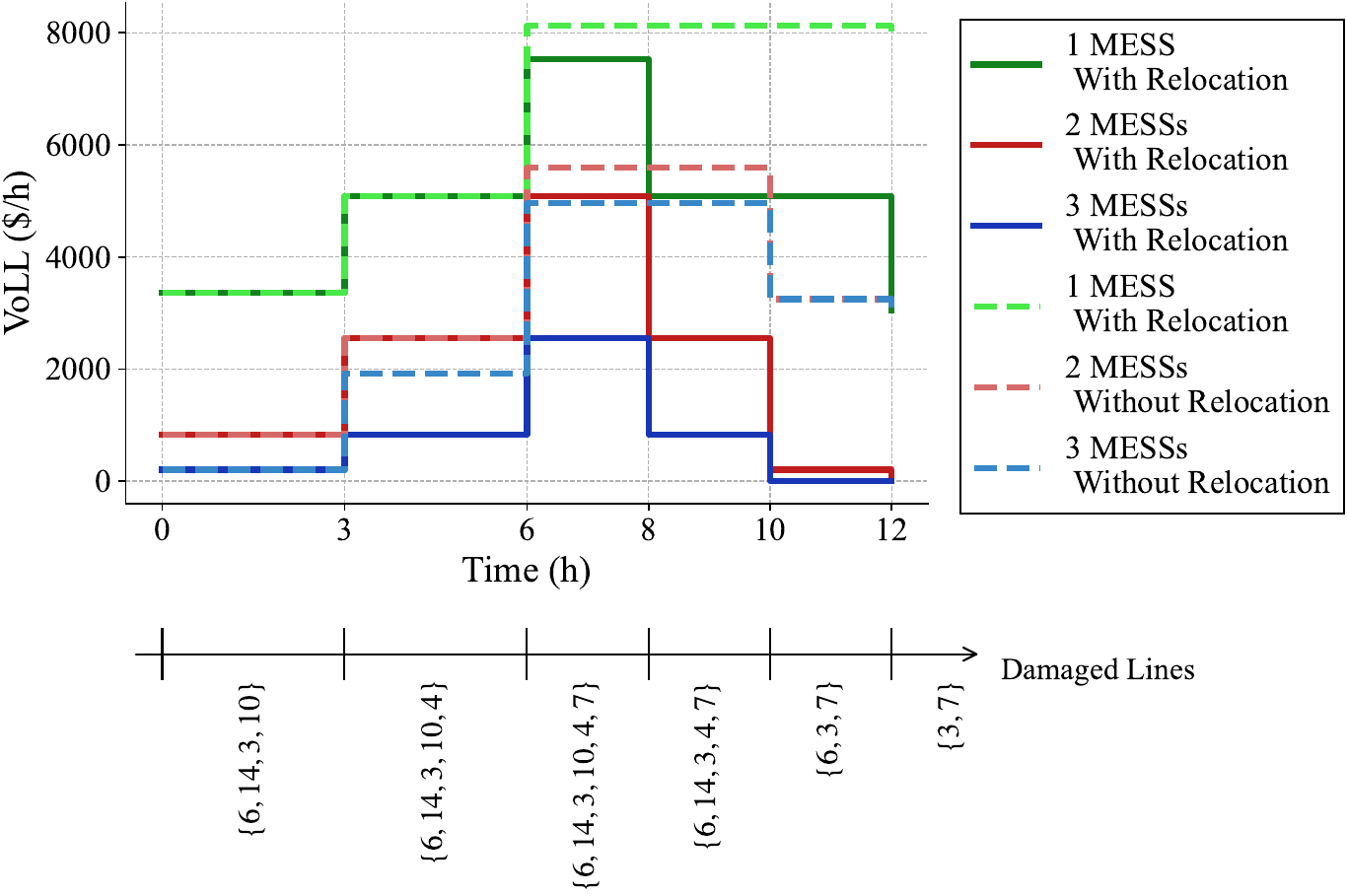}
    \caption{Comparison of the aggregated VoLL rates with and without MESS relocation during a 12-hour outage interval.}
    \label{fig:voll_evolution}
\end{figure}

\begin{table*}[t]
\centering
\caption{Grid Conditions Updates During a 12-hour Outage Interval, and Corresponding System Responses}
\label{tab:reloc_results}
\begin{tabular}{c | c | c | c | c | c | c | c | c | c | c | c}
\hline\hline
\multirow{3}{*}{Time Interval} & \multirow{3}{*}{Damaged Lines} & \multirow{3}{*}{Update Type} & \multicolumn{4}{c|}{Control Decisions} & \multicolumn{3}{c|}{\multirow{2}{*}{Energy Level}} & \multirow{3}{*}{\begin{tabular}[c]{@{}c@{}}VoLL Rate\\ (\$/h) \end{tabular}} & \multirow{3}{*}{\begin{tabular}[c]{@{}c@{}}Total CoLL \\ (\$) \end{tabular}}\\
\cline{4-7}
 & & & \multirow{2}{*}{$\sigma_d$} & \multicolumn{3}{c|}{MES Locations} & \multicolumn{3}{c|}{} & & \\
\cline{5-10}
 & & &  & $M_1$ & $M_2$ & $M_3$ & $M_1$ & $M_2$ & $M_3$ & & \\
\hline
0.0 - 4.0 & \{6,14,3,10\} & New Fault & \{10\}   & 11   & 4 & 15 & 48\%$\downarrow$ & 33\%$\downarrow$ & 63\%$\downarrow$ & 204.0 & 816.0\\
4.0 - 6.1 & \{6,14,3,10,4\} & MES Depleted & \{10,15\} & 11 & 4  & 5 & 26\%$\downarrow$ & 0\%$\downarrow$ & 35\%$\downarrow$ & 828.0 & 1,738.8\\
6.1 - 8.4 & \{6,14,3,10,4\} & MES Depleted & \{10,15,5\} & 11 & 9  & 4 & 0\%$\downarrow$ & 92\%$\uparrow$ & 8\%$\downarrow$ & 2,548.0 & 5,860.4\\
8.4 - 8.6 & \{6,14,3,10,4\} & MES Charged & \{10,15,5,4,14\} & 3 & 9  & 11 & 7\%$\uparrow$ & 100\%$\uparrow$ & 2\%$\downarrow$ & 5,085.0 & 1,017.0\\
8.6 - 8.7 & \{6,14,3,10,4\} & MES Depleted & \{10,15,5\} & 3 & 4  & 11 & 10\%$\uparrow$ & 98\%$\downarrow$ & 0\%$\downarrow$ & 2,548.0 & 254.8\\
8.7 - 10.0 & \{6,14,3,10,4\} & Line Restored & \{10,15,5,4,14\} & 3 & 11  & 6 & 61\%$\uparrow$ & 77\%$\downarrow$ & 52\%$\uparrow$ & 5,085.0 & 6,610.5\\
10.0 - 11.1 & \{6,14,10,4\} & MES Charged & \{10,5,15\} & 3 & 11  & 6 & 100\%$\uparrow$ & 68\%$\downarrow$ & 89\%$\uparrow$ & 2,548.0 & 2,802.8\\
11.1 - 11.4 & \{6,14,10,4\} & MES Charged & \{10,15\} & 5 & 11  & 6 & 95\%$\downarrow$ & 65\%$\downarrow$ & 100\%$\uparrow$ & 828.0 & 248.4\\
11.4 - 12 & \{6,14,10,4\} & Outage Resolved & \{10\} & 5 & 11  & 15 & 85\%$\downarrow$ & 59\%$\downarrow$ & 95\%$\downarrow$ & 204.0 & 122.4\\
\hline\hline
\end{tabular}
\end{table*}

\begin{table}[!t]
\caption{Aggregated Value of Lost Load (VoLL) Rates and 
Corresponding Cost of Lost Load (CoLL) Over Different Time Intervals}
\label{tab:CoLL}
\centering
\begin{tabular}{ccc}
\hline
\textbf{Time Interval (h)} & \textbf{VoLL Rate (\$/h)} & \textbf{CoLL (\$)} \\
\hline
0.0--4.0   & 204.0   & 816.0   \\
4.0--6.1   & 828.0   & 1,738.8 \\
6.1--8.4   & 2,548.0 & 5,860.4 \\
8.4--8.6   & 5,085.0 & 1,017.0 \\
8.6--8.7   & 2,548.0 & 254.8   \\
8.7--10.0  & 5,085.0 & 6,610.5 \\
10.0--11.1 & 2,548.0 & 2,802.8 \\
11.1--11.4 & 828.0   & 248.4   \\
11.4--12.0 & 204.0   & 122.4   \\
\hline
\end{tabular}
\end{table}

\begin{figure}[!t]
    \centering
    \includegraphics[width=\columnwidth]{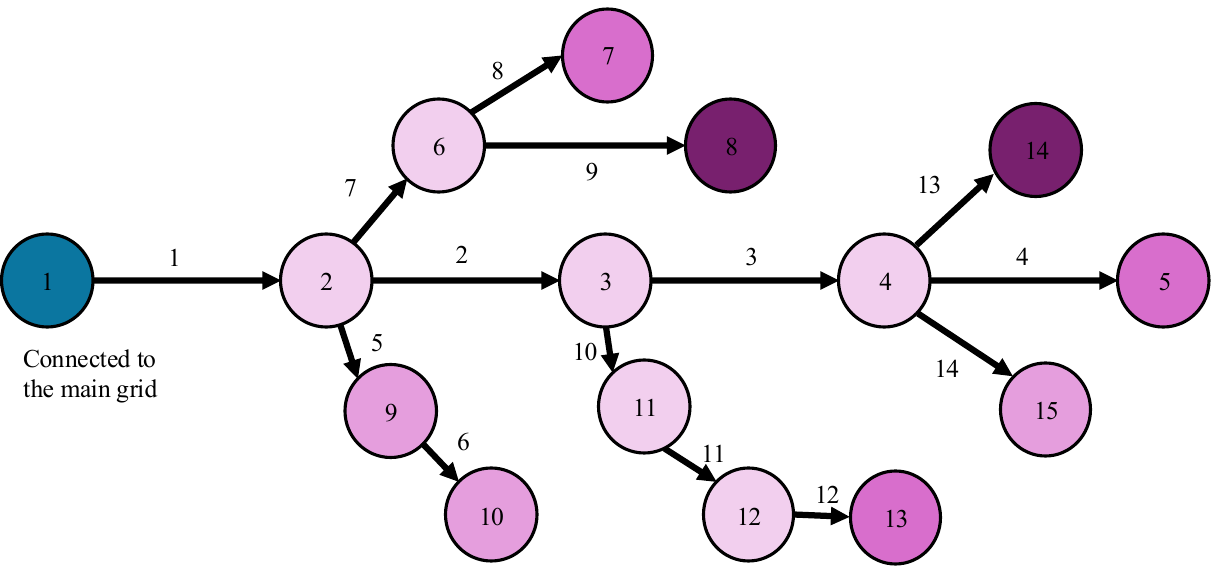}
    \caption{Nodes and lines indexing in the IEEE 15-bus system. Priority of loads, modeled by their associated VoLLs, is represented by the intensity of their color. Node "1" is assumed to be connected to the main grid.}
    \label{fig:15bus}
\end{figure}

\section{Conclusion}
\label{section: Conclusion}

This paper proposed a three-stage framework for scheduling MESSs on a coupled transportation-grid network, with the primary objective of enhancing grid resilience against extreme events.
The proposed approach leverages MESS capacity for performing energy arbitrage during normal operation, while proactively pre-staging these units on the TN prior to an outage to minimize their ETA at critical loads. To achieve this, a clustering and spatial search algorithm was developed to identify optimal pre-allocation sites on the TN.
In the post-disaster phase of the proposed framework, MESSs are dynamically relocated across the network in response to updates on grid conditions, ensuring the overall CoLL is continually minimized.
Simulation results show that we can achieve significant reduction in CoLL using the proposed framework in both pre-disaster stage, through reducing the ETA at loads, and post-disaster stage by adapting to the changes in the grid conditions while also being able to reduce the total generation cost during normal operation.

\FloatBarrier

\bibliographystyle{IEEEtran}
\bibliography{references}

@article{ahmadi2015mathematical,
  title={Mathematical representation of radiality constraint in distribution system reconfiguration problem},
  author={Ahmadi, Hamed and Mart{\'\i}, Jos{\'e} R},
  journal={Int. J. Electr. Power Energy Syst.},
  volume={64},
  pages={293--299},
  year={2015},
  publisher={Elsevier}
}

@ARTICLE{6507355,
  author={Farivar, Masoud and Low, Steven H.},
  journal={IEEE Trans. Power Syst.}, 
  title={Branch Flow Model: Relaxations and Convexification—Part I}, 
  year={2013},
  volume={28},
  number={3},
  pages={2554-2564},
  doi={10.1109/TPWRS.2013.2255317}}

@ARTICLE{6815671,
  author={Low, Steven H.},
  journal={IEEE Trans. Control Netw. Syst.}, 
  title={Convex Relaxation of Optimal Power Flow—Part II: Exactness}, 
  year={2014},
  volume={1},
  number={2},
  pages={177-189},
  doi={10.1109/TCNS.2014.2323634}}

@ARTICLE{25627,
  author={Baran, M.E. and Wu, F.F.},
  journal={IEEE Trans. Power Del.}, 
  title={Network reconfiguration in distribution systems for loss reduction and load balancing}, 
  year={1989},
  volume={4},
  number={2},
  pages={1401-1407},
  doi={10.1109/61.25627}}

@ARTICLE{7801854,
  author={Panteli, Mathaios and Pickering, Cassandra and Wilkinson, Sean and Dawson, Richard and Mancarella, Pierluigi},
  journal={IEEE Trans. Power Syst.}, 
  title={Power System Resilience to Extreme Weather: Fragility Modeling, Probabilistic Impact Assessment, and Adaptation Measures}, 
  year={2017},
  volume={32},
  number={5},
  pages={3747-3757},
  doi={10.1109/TPWRS.2016.2641463}}

@ARTICLE{9372331,
  author={Nazemi, Mostafa and Dehghanian, Payman and Lu, Xiaonan and Chen, Chen},
  journal={IEEE Trans. Smart Grid}, 
  title={Uncertainty-Aware Deployment of Mobile Energy Storage Systems for Distribution Grid Resilience}, 
  year={2021},
  volume={12},
  number={4},
  pages={3200-3214},
  doi={10.1109/TSG.2021.3064312}}

@article{ibne2020modeling,
  title={Modeling and assessing cyber resilience of smart grid using Bayesian network-based approach: A system of systems problem},
  author={Ibne Hossain, Niamat Ullah and Nagahi, Morteza and Jaradat, Raed and Shah, Chiranjibi and Buchanan, Randy and Hamilton, Michael},
  journal={J. Comput. Des. Eng.},
  volume={7},
  number={3},
  pages={352--366},
  year={2020},
  publisher={Oxford University Press}
}

@article{grotschel1989cutting,
  title={A cutting plane algorithm for a clustering problem},
  author={Gr{\"o}tschel, Martin and Wakabayashi, Yoshiko},
  journal={Math. Program.},
  volume={45},
  number={1},
  pages={59--96},
  year={1989},
  publisher={Springer}
}

@article{shi2022enhancing,
  title={Enhancing distribution system resilience against extreme weather events: Concept review, algorithm summary, and future vision},
  author={Shi, Qingxin and Liu, Wenxia and Zeng, Bo and Hui, Hongxun and Li, Fangxing},
  journal={Int. J. Electr. Power Energy Syst.},
  volume={138},
  pages={107860},
  year={2022},
  publisher={Elsevier}
}

@article{paul2024resilience,
  title={Resilience assessment and planning in power distribution systems: Past and future considerations},
  author={Paul, Shuva and Poudyal, Abodh and Poudel, Shiva and Dubey, Anamika and Wang, Zhaoyu},
  journal={Renew. Sustain. Energy Rev.},
  volume={189},
  pages={113991},
  year={2024},
  publisher={Elsevier}
}

@article{chuangpishit2023mobile,
  title={Mobile energy storage systems: A grid-edge technology to enhance reliability and resilience},
  author={Chuangpishit, Shadi and Katiraei, Farid and Chalamala, Babu and Novosel, Damir},
  journal={IEEE Power Energy Mag.},
  volume={21},
  number={2},
  pages={97--105},
  year={2023},
  publisher={IEEE}
}

@article{dugan2021application,
  title={Application of mobile energy storage for enhancing power grid resilience: A review},
  author={Dugan, Jesse and Mohagheghi, Salman and Kroposki, Benjamin},
  journal={Energies},
  volume={14},
  number={20},
  pages={6476},
  year={2021},
  publisher={MDPI}
}

@article{kim2018enhancing,
  title={Enhancing distribution system resilience with mobile energy storage and microgrids},
  author={Kim, Jip and Dvorkin, Yury},
  journal={IEEE Trans. Smart Grid},
  volume={10},
  number={5},
  pages={4996--5006},
  year={2018},
  publisher={IEEE}
}

@article{yao2019rolling,
  title={Rolling optimization of mobile energy storage fleets for resilient service restoration},
  author={Yao, Shuhan and Wang, Peng and Liu, Xiaochuan and Zhang, Huajun and Zhao, Tianyang},
  journal={IEEE Trans. Smart Grid},
  volume={11},
  number={2},
  pages={1030--1043},
  year={2019},
  publisher={IEEE}
}

@article{lu2024mobile,
  title={Mobile energy-storage technology in power grid: A review of models and applications},
  author={Lu, Zhuoxin and Xu, Xiaoyuan and Yan, Zheng and Han, Dong and Xia, Shiwei},
  journal={Sustainability},
  volume={16},
  number={16},
  pages={6857},
  year={2024},
  publisher={MDPI}
}

@article{aslam2025application,
  title={Application of energy storage systems to enhance power system resilience: A critical review},
  author={Aslam, Muhammad Usman and Miah, Md Sazal and Amin, BM Ruhul and Shah, Rakibuzzaman and Amjady, Nima},
  journal={Energies},
  volume={18},
  number={14},
  pages={3883},
  year={2025},
  publisher={MDPI}
}

@article{rajabzadeh2022improving,
  title={Improving the resilience of distribution network in coming across seismic damage using mobile battery energy storage system},
  author={Rajabzadeh, Mohammad and Kalantar, Mohsen},
  journal={J. Energy Storage},
  volume={52},
  pages={104891},
  year={2022},
  publisher={Elsevier}
}

@article{guo2023mobile,
  title={Mobile energy storage system scheduling strategy for improving the resilience of distribution networks under ice disasters},
  author={Guo, Xiaofang and Miao, Guixi and Wang, Xin and Yuan, Liang and Ma, Hengrui and Wang, Bo},
  journal={Processes},
  volume={11},
  number={12},
  pages={3339},
  year={2023},
  publisher={MDPI}
}

@article{zhou2024bi,
  title={A bi-level mobile energy storage pre-positioning method for distribution network coupled with transportation network against typhoon disaster},
  author={Zhou, Ke and Jin, Qingren and Feng, Bin and Wu, Lifang},
  journal={IET Renew. Power Gener.},
  volume={18},
  number={16},
  pages={3776--3787},
  year={2024},
  publisher={Wiley Online Library}
}

@article{hua2023robust,
  title={Robust emergency preparedness planning for resilience enhancement of energy-transportation nexus against extreme rainfalls},
  author={Hua, Zhihao and Zhou, Bin and Or, Siu Wing and Zhang, Jie and Li, Canbing and Wei, Juan},
  journal={IEEE Trans. Ind. Appl.},
  volume={60},
  number={1},
  pages={1196--1207},
  year={2023},
  publisher={IEEE}
}

@article{shen2023mobile,
  title={Mobile energy storage systems with spatial--temporal flexibility for post-disaster recovery of power distribution systems: A bilevel optimization approach},
  author={Shen, Yueqing and Qian, Tong and Li, Weiwei and Zhao, Wei and Tang, Wenhu and Chen, Xingyu and Yu, Zeyuan},
  journal={Energy},
  volume={282},
  pages={128300},
  year={2023},
  publisher={Elsevier}
}

@ARTICLE{10311549,
  author={Chen, Hongzhou and Xiong, Xiaofu and Zhu, Jizhong and Wang, Jian and Wang, Wei and He, Yufei},
  journal={IEEE Trans. Power Del.}, 
  title={A Two-Stage Stochastic Programming Model for Resilience Enhancement of Active Distribution Networks With Mobile Energy Storage Systems}, 
  year={2024},
  volume={39},
  number={4},
  pages={2001-2014},
  doi={10.1109/TPWRD.2023.3321062}
}

@ARTICLE{9634033,
  author={Lu, Zhuoxin and Xu, Xiaoyuan and Yan, Zheng and Shahidehpour, Mohammad},
  journal={IEEE Trans. Transp. Electrific.}, 
  title={Multistage Robust Optimization of Routing and Scheduling of Mobile Energy Storage in Coupled Transportation and Power Distribution Networks}, 
  year={2022},
  volume={8},
  number={2},
  pages={2583-2594},
  doi={10.1109/TTE.2021.3132533}
}

@Article{app142210367,
AUTHOR = {Lei, Hongtao and Jiang, Bo and Liu, Yajie and Zhu, Cheng and Zhang, Tao},
TITLE = {Two-Stage Optimization of Mobile Energy Storage Sizing, Pre-Positioning, and Re-Allocation for Resilient Networked Microgrids with Dynamic Boundaries},
JOURNAL = {Appl. Sci.},
VOLUME = {14},
YEAR = {2024},
NUMBER = {22},
ARTICLE-NUMBER = {10367},
URL = {https://www.mdpi.com/2076-3417/14/22/10367},
ISSN = {2076-3417},
DOI = {10.3390/app142210367}
}

@ARTICLE{10049743,
  author={Zhao, Zehua and Luo, Fengji and Zhu, Jizhong and Ranzi, Gianluca},
  journal={IEEE Trans. Smart Grid}, 
  title={Multi-Stage Mobile BESS Operational Framework to Residential Customers in Planned Outages}, 
  year={2023},
  volume={14},
  number={5},
  pages={3640-3653},
  doi={10.1109/TSG.2023.3245094}
}

@article{battaglia2018relational,
  title={Relational inductive biases, deep learning, and graph networks},
  author={Battaglia, Peter W and Hamrick, Jessica B and Bapst, Victor and Sanchez-Gonzalez, Alvaro and Zambaldi, Vinicius and Malinowski, Mateusz and Tacchetti, Andrea and Raposo, David and Santoro, Adam and Faulkner, Ryan and others},
  journal={arXiv:1806.01261},
  volume={10},
  year={2018}
}

@InProceedings{pmlr-v70-gilmer17a,
  title = 	 {Neural Message Passing for Quantum Chemistry},
  author =       {Justin Gilmer and Samuel S. Schoenholz and Patrick F. Riley and Oriol Vinyals and George E. Dahl},
  booktitle = 	 {Proc. 34th Int. Conf. Mach. Learn.},
  pages = 	 {1263--1272},
  year = 	 {2017},
  editor = 	 {Precup, Doina and Teh, Yee Whye},
  volume = 	 {70},
  series = 	 {Proc. Mach. Learn. Res.},
  month = 	 {06--11 Aug},
  publisher =    {PMLR},
}

\end{document}